\documentclass{amsart}

\usepackage{float}
\usepackage{color}
\usepackage{xypic}
\usepackage{enumerate}
\usepackage{tikz}
\usetikzlibrary{calc}

\def\k{\mathfrak k}      
\def\n{\mathfrak n}      

\def\r{\mathfrak r}

\def\c{\mathbf c}

\def\Re{\mathrm{Re}}  

\newtheorem{theorem}{Theorem}[section]
\newtheorem{lemma}{Lemma}
\newtheorem{corollary}{Corollary}
\newtheorem{proposition}{Proposition}
\theoremstyle{conjecture}

\theoremstyle{definition}

\theoremstyle{question}
\newtheorem*{question}{Question}
\theoremstyle{questions}

\theoremstyle{examples}
\newtheorem*{examples}{Examples}

\theoremstyle{remark}
\newtheorem{remark}{Remark}
\theoremstyle{remarks}
\newtheorem*{remarks}{Remarks}
\theoremstyle{example}

\numberwithin{equation}{section}

\begin{document}

\title{Conversations with Flaschka: Kac-Moody groups and Verblunsky coefficients}

\author{Mohammad Javad Latifi}
\email{mohammad.javad.latifi.jebelli@dartmouth.edu}
\address{Department of Mathematics, Dartmouth College, 
Hanover, NH 03755}

\author{Doug Pickrell}
\email{pickrell@math.arizona.edu}

\address{Mathematics Department, University of Arizona, Tucson, AZ 85721}

\begin{abstract}In this paper we discuss two items which in one way or another originated from conversations
with Hermann Flaschka and his students. The first is an application of the Toda lattice to the question of
whether there exists a complex Lie group (with certain properties) associated to an indefinite type Kac-Moody Lie algebra. The second concerns a new example of the Verblunsky correspondence.
\end{abstract}
\maketitle

\setcounter{section}{-1}

\section{Preface}

This paper is dedicated to the memory of Hermann Flaschka. Hermann made seminal contributions
to completely integrable systems and many related areas of mathematics, reflecting his broad interests.
He was beloved as a colleague
and friend for his humility, wisdom and wit. His intensity in pursuit of facts and solutions, whether related
to mathematical problems or raises for adjuncts or climate change, was legendary.

Many decades ago, after a sanitized introduction to abstract mathematics, I (the second author)
took a course from Hermann which emphasized explicit formulas for tau functions and solitons. This was a formative experience. A decade ago Hermann, Bole Yang (Hermann's last PhD student) and I spent
a semester learning about the Verblunsky correspondence, CMV matrices and the discrete nonlinear Schroedinger
equation, using the analogy with Jacobi matrices and the Toda lattice as a guide.

\section{Introduction}

The goal of this paper is to illustrate some ideas which in one way or another originated from some conversations
with Hermann Flaschka and his students.

\subsection{Kac-Moody Groups}

Kac-Moody algebras fall into three classes: finite dimensional,
e.g. $sl(n,\mathbb C)$, which is the setting for Flaschka's original reformulation of the nonperiodic
Toda equations; affine, e.g. $sl(n,\mathbb C[z,z^{-1}]))$ (matrices with finite Laurent series coefficients),
which for example are relevant to the periodic Toda equations; and indefinite.
Kac and Peterson associated an `algebraic group' to a general Kac-Moody algebra. For the affine and indefinite classes, this algebraic group lacks
an exponential map, hence it is natural to ask whether there exists a Lie group (with an exponential map)
which is associated to a completion of the algebra, e.g. $C^0$ or $C^1$ or $C^{\infty}(S^1,SL(n,\mathbb C))$ in the affine case. Toda equations can be associated to the generalized Cartan matrix for any Kac-Moody algebra,
and in conventional cases solving the equations involves exponentiation and LDU type factorization. Our main point is that properties of the Toda equation solutions are intimately related to the question of whether
there exists a complex Lie group with exponential map associated to a given Kac-Moody algebra. The dissertation related work of Hermann's student Maorong Zou suggests that there does not exist a complex Lie group, having certain properties, in the indefinite cases, but to our knowledge there does not exist a definitive theorem.

\begin{remark}In his work on moonshine conjectures and the Monster finite simple group, Borcherds enlarged an indefinite type Lie algebra to the so called Monster Lie algebra (see \cite{Goddard} for a short introduction).
For recent work on the question of how to associate groups to this algebra and others, see \cite{CJM} and the references.
\end{remark}

\subsection{The Verblunsky Correspondence}

Let $Prob(S^1)$ denote the metrizable compact convex set of probability measures on
the unit circle $S^1$, let $Prob'(S^1)$ denote the subset
consisting of measures with infinite support, and let $\Delta:=\{z\in\mathbb C:|z|<1\}$. The
Verblunsky correspondence refers to a miraculous homeomorphism
\begin{equation}\label{verblunskymap}Prob'(S^1) \leftrightarrow
\prod_{n=1}^{\infty}\Delta: \mu \leftrightarrow \alpha\end{equation}
or more generally to a homeomorphism of $Prob(S^1)$ and a compactification of the product space.
This correspondence can be described in several different ways. One formulation is as follows.
Given $\mu\in Prob'(S^1)$, if $p_0 =z^0$, $p_1(z)$, $p_2(z)$,... are the monic orthogonal polynomials corresponding
to $\mu$, then $\alpha_n =p_{n}(0)^*$ (where $(\cdot)^*$ is complex conjugation). Conversely, given $\alpha$, $\mu$ is the weak* limit of probability measures
\begin{equation}\label{verblunskymap2}d\mu= \lim_{N\to\infty}
\frac{\prod_{n=1}^{N}(1-\vert \alpha_n\vert^2)}{\vert p_N(z)\vert^2}\frac{d\theta}{2\pi}\end{equation}
where $p_0=1$,
\begin{equation}\label{Szego} p_{n}(z) = zp_{n-1}(z)+\alpha_n^*z^{n-1}p_{n-1}^{*}(z),\quad  n >0,\end{equation} and
$p_n^*(z)=p_n(\frac1{z^*})^*$.

\begin{remarks} (a) This correspondence is truly `miraculous' because it yields a parametrization of (nontrivial) measures on the circle by parameters which are functionally independent - this has consequences, see e.g. \cite{LP}.

(b) For a masterful survey of the history and ramifications of the Verblunsky theory,
see Simon's two volume work, \cite{Simon1}, or for a short survey, see \cite{Simon2}. For more recent connections with completely integrable systems, see \cite{Nenciu} and \cite{Yang}.
\end{remarks}

There are a number of known explicit examples of this correspondence, see section 1.6 of Volume I of \cite{Simon1}, including the case $\beta=1$ of the following

\begin{theorem} \label{thm1} Suppose that $\beta>0$. The probability measure $\mu=\frac{1}{\mathfrak Z}(1-cos(\theta))^{1/\beta}\frac{d\theta}{2\pi}$ corresponds to the sequence
$\alpha_n=\frac{1}{1+n\beta}$, $n=1,2,...$.
\end{theorem}

This example illuminates the relation between decay properties of the coefficients and regularity properties of the corresponding measures, which we will highlight in the text.

It is a simple matter to use the definition of the Verblunsky coefficients to compute a few of the coefficients; presumably we could have discovered this example using these calculations.  However for us this arose from mulling the following formula,
\begin{equation}\label{niceidentity}\sum_{i,j} \prod_{u=1}^L\frac{1}{(i(u)\beta+1)(j(u)\beta+1)}=\prod_{k=1}^{n}(\frac{1}{k}\beta^{-1}+\frac{k-1}{k})  \end{equation}
where the sum is over all integral sequences $i(1)>j(1)>...>i(L)>j(L)\ge 0$ satisfying $\sum_{u=1}^L(i(u)-j(u))=n$,
for some $1\le L\le n$ (see \cite{LP}). This leads to `super-telescoping formula' (see \cite{Latifi_Lattice})

\begin{equation}\sum_{n=0}^{\infty} \left( \sum_{i,j} \prod_{u=1}^L\frac{1}{(i(u)\beta+1)(j(u)\beta+1)} \right) z^n= \frac{1}{(1-z)^{1/\beta}}
\end{equation}

It turns out that inverting the Verblunsky map is related to `root subgroup factorization' (for an introduction to root subgroup factorization, see \cite{BP}). Theorem \ref{thm1}
is essentially a corollary of the following

\begin{theorem}\label{thm2} Suppose that $\alpha_n=1/(1+n\beta)$ as before and let $m=1/\beta$. Then
\begin{equation}\lim_{n\to\infty}\left(\begin{matrix} 1&\alpha_nz^{-n}\\
\bar{\alpha}_nz^n&1\end{matrix} \right)...\left(\begin{matrix} 1&
\alpha_1z^{-1}\\
\bar{\alpha}_1z&1\end{matrix} \right)=\left(\begin{matrix} \frac{1+F_m^*(z)}{2(1-z^*)^m}&\frac{1-F_m^*(z)}{2(1-z^*)^m} \\\frac{1-F_m(z)}{2(1-z)^m}&\frac{1+F_m(z)}{2(1-z)^m}\end{matrix}\right) \end{equation}
where
$$F_m(z)=1+\sqrt{\pi}2^{1-m}\frac{\Gamma(m+ 1)}{\Gamma(m+\frac 12)}\sum_{n=1}^{\infty}\sum_{k=0}^{\infty}\left(\begin{matrix}m\\n+2k
\end{matrix}\right)\left(\begin{matrix}n+2k\\k\end{matrix}\right)(\frac{-1}{2})^{n+2k}z^{n} $$
(If $m$ is an integer, this is a polynomial; in general the binomial coefficient is defined in terms
of gamma functions).
\end{theorem}

One interesting thing to observe is that this infinite product is not bounded, reflecting the non-summability of the coefficients. This has consequences.

\subsection{Plan of the Paper}

The prerequisites for Section \ref{KMtheory} are steep. Our main goals are to given an accessible and motivated introduction to questions about groups associated to Kac-Moody algebras, and to explain why properties of generalized Toda equations are inextricably linked to questions about Kac-Moody groups.

Section \ref{Verblunsky} is more straightforward. Our main purpose is to prove Theorems \ref{thm1} and \ref{thm2},
and some related calculations.

\subsection{Notation}\label{notation}
If $f(z)=\sum f_nz^n$, then we will write
$$f=f_-+f_0+f_+$$
where $f_-(z)=\sum_{n<0}f_nz^n$ and
$f_+(z)=\sum_{n>0}f_nz^n$, $(f)_{0+}=f_0+f_+$, and $f^*(z)=\sum (f_{-n})^*z^n$,
where $w^*=\bar{w}$ is the complex
conjugate of the complex number $w$.
If the Fourier series is convergent at a point $z\in S^1$, then $f^*(z)$ is the conjugate of the complex number $f(z)$.  If $f \in
H^0(\Delta)$, then $f^* \in H^0(\Delta^*)$, where $\Delta$ is the
open unit disk, $\Delta^*$ is the open unit disk at $\infty$, and
$H^0(U)$ denotes the space of holomorphic functions for a domain $U\subset \mathbb C$.

Our conventions regarding Verblunsky coefficients differ from those in \cite{Simon1}, where the nontrivial coefficients are indexed by $n=0,1,...$.
We use the convention that $\alpha_0=1$ and the nontrivial Verblunsky coefficients are $\alpha_1,\alpha_2,...$.

\section{An Application of Toda Lattice Theory to Kac-Moody Groups}\label{KMtheory}

Chevalley and Serre famously discovered that a finite dimensional simple complex
Lie algebra can be defined in terms of generators and relations using
a  Cartan matrix $A$, i.e. an integral $n\times n$ matrix having
positive principal minors and
satisfying $A_{ii}=2$, $A_{ij}\le 0$  and $A_{ij}=0 \implies A_{ji}=0$
for $i\ne j$. Kac and Moody independently deleted the condition that
$A$ has positive principal minors, and discovered that the theory remained robust
(one can also allow $n$ to be infinite, which is essential in Borcherds
work on the Monster). The corresponding Kac-Moody algebra
$\mathfrak g=\mathfrak g(A)$ is
generated by $3n$ letters $e_i,$ $f_i,$ and $h_i$ $(1\le i<n+1)$ satisfying the
Chevalley-Serre relations \newline
\centerline{$[h_i,h_j]=0,$ $[e_i,f_i$ $]=h_i$, $[e_i$ $,f_j$ $]=0
,$     $i\ne j$}
\centerline{$[h_i$ $,e_j]=A_{ij}e_j$, $[h_i,f_j]=-A_{ij}f_j,$}
\centerline{$(ad(e_i))^{1-A_{ij}}e_j=0,$ $(ad(f_i))^{1-A_{ij}}f_j
=0,$     $i\ne j$   }
The main take away is this: for each $i=1,...,n$, there is an isomorphism
$sl(2,\mathbb C)\leftrightarrow span(f_i,h_i,e_i)$,
\begin{equation}\label{sl2c} \left(\begin{matrix}0&0\\1&0\end{matrix}\right) \leftrightarrow f_i,\quad \left(\begin{matrix}1&0\\0&-1\end{matrix}\right) \leftrightarrow h_i,\quad\left(\begin{matrix}0&1\\0&0\end{matrix}\right) \leftrightarrow e_i
\end{equation}
and the other relations detail how these $n$ copies of $sl(2,\mathbb C)$ (the `root subalgebras' corresponding to simple roots) interact.
A technical aside: We will assume that $A$ is irreducible and
symmetrizable, meaning that there exists an essentially unique
nondegenerate symmetric bilinear form on $\mathfrak g$, denoted $\langle \cdot,\cdot\rangle$.

As in the finite dimensional case, the Lie algebra $\mathfrak g$ has a triangular decomposition\newline
\centerline{$\mathfrak g=\mathfrak n^{-}\oplus \mathfrak h\oplus \mathfrak n^{+}$}
where $\mathfrak n^{-}$ is generated by $\{f_j\}$ (think lower triangular matrices, or creation operators), $\mathfrak h$ by $
\{h_j\}$ (think diagonal matrices), and $\mathfrak n^{+}$
by $\{e_j\}$ (think upper triangular matrices, or annihilation operators).

As we noted in the introduction, Kac-Moody Lie algebras fall into three classes: (1) $\mathfrak g$ is finite dimensional, in which
case there exists a corresponding complex Lie group $G(A)$. (2) $\mathfrak g$ is of affine type, in which
case there exists a finite dimensional simple complex Lie algebra $\dot{\mathfrak g}$ such that
$\mathfrak g$ is a central extension of the loop algebra of functions $S^1 \to \dot{\mathfrak g}$
having finite Fourier series (or a twisted version of this); in this case there exist infinite dimensional
Lie groups corresponding to completions of $\mathfrak g$, e.g. (a central extension of)
the loop group $C^{\infty}(S^1,\dot{G}(A))$ (see \cite{PS}). And finally
(3) $\mathfrak g$ is of indefinite type.

\begin{question}In the indefinite case, does there exist a corresponding complex Lie group?
\end{question}

This is an interesting question, because for example in conventional situations one uses the
global group structure to solve associated differential equations, such as the Toda equations.

Kac and Peterson showed in general how to construct an `algebraic group' $G(A)$ which is associated to $\mathfrak g(A)$ (see \cite{Kac} and \cite{KP}). The basic idea is simple: $G(A)$ is generated by $n$ `root subgroups', i.e. subgroups isomorphic to $SL(2,\mathbb C)$ corresponding to the $n$ copies of $sl(2,\mathbb C)$ in (\ref{sl2c})). One imposes all the relations that arise from realizing these groups in a highest weight representation of the Lie algebra (meaning that the $f_i$ are creation operators and the $e_i$ are annihilation operators). The paper \cite{KP} is a masterful analysis of the essential relations, which in the words of the authors are 'analytic continuations of the relations defining the Weyl group', which are Coxeter group relations.

$G(A)$ is not a Lie group, except in the finite case. $G(A)$ is useful in formulating global algebraic notions, such as LDU factorization: for $g\in G(A)$, $g=ldu$,
where $l\in N^-$, $d\in H$, $u\in N^+$ (the groups corresponding to $\mathfrak n^-,\mathfrak h,\mathfrak n^+$, respectively) if and only if $\sigma_j(g)\ne 0$, where the $\sigma_j$ are the fundamental matrix coefficients, $\sigma_j(g)=\langle g\cdot v,v\rangle$ where $v$ is the vacuum for the $i$th fundamental highest weight representation (For $g\in SL(n,\mathbb C)$ this is the statement that $g$ has a LDU factorization iff the principal minors are nonsingular, and for $g\in SL(n,\mathbb C[z,z^{-1}])$ this is essentially the statement that g has a Riemann-Hilbert factorization (with factors having finite Laurent expansions, a nontrivial condition) iff the associated Toeplitz operator is invertible; in all cases $d=\prod_{j=1}^n \sigma_j(g)^{h_j}$, and - what is essential for us - the $l$ and $u$ factors are rational in any given representation). But the ordinary exponential function is transcendental, and $G(A)$ does not have an exponential map, except in the finite case. We would like to complete $G(A)$ so that it has an exponential map, while preserving its algebraic properties.

To formalize this, assume that there exists a complex Lie group $\overline G$ which has a Lie algebra $\overline{\mathfrak g}$ containing $\mathfrak g$ as a dense subalgebra. Without further comment we will denote the completion of $\mathfrak n^-$ by
$\overline{\mathfrak n^-}$ and so on. We will assume that $\overline G$ has the following properties. First, the form $\langle\cdot,\cdot \rangle$ extends continuously in each variable to $\overline {\mathfrak g}$; secondly,
the exponential map $exp:\overline{\mathfrak g} \to \overline{G}$ is well-defined and holomorphic; and thirdly, $\overline G$ has a Birkhoff stratification parameterized by the elements of the Weyl group $W=N(H)/H$, compatible
with, and having the same basic properties as, the Birkhoff stratification for $G(A)$. We will not spell out the third condition in full detail, but simply note $g$ will be in the top stratum $\Sigma_1$ (a dense open set) iff the holomorphic functions $\sigma_j(g)\ne 0$
iff $g$ has a unique LDU factorization, $g=ldu$, $l\in \overline{N^-}$, $d=\prod_j\sigma_j(g)^{h_j}\in H$, $u\in \overline{N^+}$, and analogous to the $G(A)$ case, $l$ and $u$ are meromorphic in any given representation.

We now mimic (what is usually referred to as) the Kostant formulation of the Toda lattice equations (see \cite{Kostant}, although many other sources could be cited).
Define $\mathfrak b^{\pm}=\mathfrak h+\mathfrak n^{\pm}$.
Since the form $\langle\cdot,\cdot\rangle$ extends continuously to $\overline {\mathfrak g}$ there is an induced map with dense image,
$$\overline{\mathfrak b^-} \to (\overline{\mathfrak b^+})^*:x \to \langle x,\cdot\rangle$$
This induces a Poisson structure on $\overline{\mathfrak b^-}$.
Define $\epsilon=\sum_{i=1}^n e_i$ and $ M=\epsilon+\overline{\mathfrak b^-} $.
Via the identification
$$\overline{\mathfrak b^-} \to M: x \to \epsilon +x$$
there is an induced Poisson structure on $M$.

Suppose that $x(t)$ is a curve in $M$. The complex Kostant version of the Toda equations is given by
\begin{equation}\label{todaeqn}\frac{d}{dt} x(t)=[proj_{\mathfrak n^-}(x(t)),x(t)]\end{equation}
Since $x=\epsilon+x_{0-}$, $x_{0-}\in \overline{\mathfrak{b}^-}$, the right hand side of (\ref{todaeqn})
is in $\overline{\mathfrak{ b}^-}$, the tangent space to $M$. The basic facts about these equations for finite type Cartan matrices carry over directly to the more general setting, namely

\begin{proposition}(a) The Toda equations are Hamiltonian for $M$. More precisely, define
$$H:M\to \mathbb C:x \to \frac 12 \langle x,x\rangle$$
Then for any smooth function $f:M\to \mathbb C $,
$$\frac{d}{dt} f(x(t))=\{H,f\}|_{x(t)} $$
for any solution $x(t)$.

(b) Given an initial condition $x_0$ with $e^{x_0}\in \Sigma_1$,
$x(t)=Ad(l(t)^{-1})(x_0)$ is a solution of (\ref{todaeqn}), up to the time when $e^{tx_0}$ exits $\Sigma_1$,
where $e^{tx_0}$ has LDU factorization $e^{tx_0}=l(t)d(t)u(t)$. $x(t)$ is a meromorphic function of $t$.

(c) Suppose that $x_0=\epsilon+\sum_{j=1}^n(a_jh_j+b_jf_j)$ (note in the $sl(n,\mathbb C)$ case this is tridiagonal).
Then the Toda equations are equivalent to
$$ \frac{d}{dt}a_j=b_j\qquad \frac{d}{dt} b_j=\sum_{i=1}^n a_i A_{ij} b_j, \qquad j=1,...,n$$
\end{proposition}

We do not need (a), although we will mention its relevance below.

\begin{proof} Define $x(t)=l(t)^{-1}x_0 l(t)$ (in some matrix representation), as in (b).
Directly from this definition one calculates that
\begin{equation}\label{formula0}\frac{d}{dt} x=[l^{-1}\frac{d}{dt} l,x]\end{equation}
Since  $ldu=e^{tx_0}$, it follows that
$$\frac{d}{dt} l du+l\frac{d}{dt}{du}=e^{tx_0} x_0$$
This implies
\begin{equation}\label{formula1}l^{-1}\frac{d}{dt} l+\frac{d}{dt}(du)(du)^{-1}=l^{-1}e^{tx_0}x_0(du)^{-1}=l^{-1}x_0e^{tx_0}(du)^{-1}\end{equation}
The last two expressions simplify to
\begin{equation}\label{formula2}=(du)x_0(du)^{-1}=l^{-1}x_0l\end{equation}
and this is equal to $x(t)$. The equality of $x$ to the first expression in \ref{formula1} implies
that
$$proj_{\mathfrak n^-}(x)=l^{-1}\frac{d}{dt} l$$
Together with (\ref{formula0}), this proves part (b).

Now consider part (c). It is convenient to write $x_0=\epsilon+h\vec a(0)+f\vec b(0)$, where
$h=(h_1,...,h_n)$, $f=(f_1,...,f_n)$ and $a$ and $b$ are columns. We also need a bit more structure.
The Lie algebra $\mathfrak g$ is graded by height, where $f_j,h_j,e_j$ have heights $-1,0,+1$ respectively.
For example $\epsilon$ is homogeneous of height $+1$, and $x_0$ is `tridiagonal'. Since $l\in \overline N^-$ and $du\in \overline B^+$, the first formula for $x(t)$ in (\ref{formula2}) implies the homogeneous components of $x(t)$ have heights $\ge -1$, and the second formula in (\ref{formula2}) implies the homogeneous components of $x(t)$ have heights $\le +1$. Thus $x(t)$ is `tridiagonal', and since $x(t)\in M$, $x(t)=\epsilon+h\vec a(t)+f\vec b(t)$.
The right hand side of (\ref{todaeqn}) is
\begin{equation}\label{todaeqn2}=[f\vec b,\epsilon+h\vec a+f\vec b] =h\vec b+[f\vec b,h\vec a]=h\vec b+f\vec{(\sum_ia_iA_{ij}b_j)}\end{equation}
where the second equality uses the Chevalley-Serre relations. This implies (c).

\end{proof}

\begin{corollary} If the solutions in (c) are not meromorphic, then there does not exist a Lie completion of $G(A)$ having the properties that we ascribed to $\overline G$ above.\end{corollary}

The Toda equations in (c) were studied by Zou in the simplest rank 2 cases, i.e. when the Cartan matrix
$A=\left(\begin{matrix}2&-k \\-l &2\end{matrix}\right)$, see chapter 5 of \cite{Zou}, who was mainly interested in integrability criteria from the Hamiltonian point of view.

\begin{remark} To put Zou's equations (5.15)-(5.18) into our form, viewing $a$ and $b$ as columns, make the linear change
$$y=\left(\begin{matrix}\frac{2}{A}&\frac{2B}{B^2+C^2}\\0&\frac{2C}{B^2+C^2}\end{matrix}\right)a,\quad
x=-\left(\begin{matrix}\frac{2}{A^2}&0\\0&\frac{2}{B^2+C^2}\end{matrix}\right)b $$
The `Cartan matrix' is then $\left(\begin{matrix}2&\frac{2AB}{B^2+C^2}\\\frac{2B}{A}&2\end{matrix}\right)$. The off diagonal entries are not necessarily nonnegative integers, so Zou is considering more general equations.

\end{remark}

In connection with this work, Zou found that the solutions are not generally meromorphic, although in retrospect a definitive theorem seems to be lacking. This appears to be a practical empirical criterion in particular cases, but to our knowledge there are not any general theorems.

\begin{remarks} (a)  The Virasoro algebra is a simple Lie algebra which is very similar to a Kac-Moody Lie algebra and which does not have a global complexification (see Proposition 3.3.2 of \cite{PS}).

(b) Physicists have been speculating for a long time that hyperbolic type Lie algebras might be concretely related to supergravity. A list of hyperbolic algebras, and some references to the physics literature, can be found in \cite{Carbone}.

\end{remarks}

\section{An Example of the Verblunsky Correspondence}\label{Verblunsky}

In the introduction we recalled one way to define the Verblunsky map $Prob'(S^1) \to \prod_{n=1}^{\infty}\Delta$, using orthogonal polynomials and the Szego recursion. This definition may seem to lack motivation, and truthfully for us the definition seems a bit ad hoc. As it turns out, there are other definitions, very different in appearance, such as the Schur algorithm, see the remarkable theorems of Geronimus in section 3 \cite{Simon2}). Hence the definition is fully justified, even if it does not seem entirely natural.

Here, in order to prove Theorems \ref{thm1} and \ref{thm2} of the introduction, we will briefly recall from Section 2 of \cite{LP} our perspective on how to invert the Verblunsky map, which emphasizes the connection with so called root subgroup factorization for the loop group $LSU(1,1)$ (see \cite{BP} for how this can be motivated).

Suppose that $\alpha\in \prod_{n=1}^{\infty}\Delta$, and suppose initially that $\sum_{n=1}^{\infty}|\alpha_n|<\infty$. Then
$$g_2(z):=\prod_{k=1}^{\stackrel{\infty}{\leftarrow}}(1-|\alpha_k|^2)^{-1/2}\left(\begin{matrix}1&\alpha_k^*z^{-k}\\ \alpha_kz^{k} & 1 \end{matrix}\right)$$
is a continuous function $g_2:S^1\to SU(1,1)$ (This is an example of a (special) root subgroup factorization; the name arises because the factors correspond to root subgroups of the loop group). This $SU(1,1)$ loop has the very special form
\begin{equation}\label{property1}g_2(z)= \left(\begin{matrix}d_2^*(z)&c_2^*(z)\\c_2(z)&d_2(z)\end{matrix}\right)\end{equation}
where $c_2$ and $d_2$ are (boundary values of) holomorphic functions in $\Delta$ satisfying $c_2(0)=0$ and $d_2(0)=1$. This loop also has a triangular factorization (essentially a Riemann-Hilbert factorization) of the especially simple form
\begin{equation}\label{property2}g_2(z)=\left(\begin{matrix}1& \bf x^*(z)\\0&1\end{matrix}\right)\left(\begin{matrix}\mathbf a_2&0\\0&\mathbf a_2^{-1}\end{matrix}\right)\left(\begin{matrix}\alpha_2(z)&\beta_2(z)\\\gamma_2(z)&\delta_2(z)\end{matrix}\right) \end{equation}
where the first matrix is holomorphic in $\Delta^*$ and $=1$ at $\infty$, $\mathbf a_2$ is a positive constant, and the third matrix is a holomorphic map $\overline{\Delta} \to SL(2,\mathbb C))$ and unipotent upper triangular at $z=0$ (see \cite{BP}). It is important to note that this depends on the continuity of $g_2$ (so that the Toeplitz operator
corresponding to this symbol is Fredholm), in addition to its special form, hence on the assumption that $\alpha$ is absolutely summable.

\begin{remark}\label{KdV} As an aside regarding completely integrable systems, in the Segal-Wilson framework
for KdV, loops into $SL(2,\mathbb C)$ map to subspaces of
a Grassmannian, which in turn correspond to solutions of the KdV equation. In the notation of \cite{SW},
if $W$ denotes the subspace corresponding to
the loop $g_2(\alpha)$ then the (initial condition for) the Baker function is
$$\psi_W(z)=1+z(x^*(z^2))=1+x_1^*z^{-1}+x_2^*z^{-3}+x_3^*z^{-5}+...$$
Thus $x^*(z)$ is essentially the initial condition for the Baker function (which itself has dependence on time).
Whether root subgroup factorization (as in \cite{BP}) has some use in the Segal-Wilson framework is unknown to us.

\end{remark}

The following is explained in section 2.1 of \cite{LP}:

\begin{proposition}\label{Verblunsky} Via the Verblunsky correspondence, the measure $\mu\in Prob'(S^1)$ that
corresponds to $\alpha\in \prod_{n=1}^{\infty}\Delta$, assuming $(\alpha_n)$ is absolutely summable, is given by
$$d\mu=\frac{\prod_{k=1}^{\infty}(1-|\alpha_k|^2)}{|\gamma_2+\delta_2|^2}\frac{d\theta}{2\pi}$$

\end{proposition}

The absolute summability of $\alpha$ guarantees that the denominator for the density of $d\mu$ is nonvanishing around the circle, hence the density is continuous and invertible. Theorem \ref{thm1} will show this fails for the marginally non-summable sequence $\alpha_n=\frac{1}{1+n\beta}$. In general given $\alpha\in \prod \Delta$, $\alpha$ and $\delta$ are holomorphic in $\Delta$, and the sum $\alpha+\delta$ is nonvanishing in $\Delta$ (see the Caratheodory function
in the next paragraph). It is the boundary behavior which is in question.

The measure $\mu$ is determined by a number of ancillary objects which are explained in Section 2 of \cite{Simon2}.
From our perspective on the inverse of the Verblunsky map, these are expressed in the following way.
Because $g_2$ has values in $SU(1,1)$, the so called Schur function
$$\mathbf f(z):=-\frac{c_2(z)}{zd_2(z)}=-\frac{\gamma_2(z)}{z\delta_2(z)}$$
is a holomorphic map $\Delta \to \Delta$ (this follows from $|d_2|^2-|c_2|^2=1$ on $S^1$, $c_2(0)=0$, and the maximum principle). The Caratheodory function, implicitly defined
by $F(z)=\frac{1+z\mathbf f(z)}{1-z\mathbf f(z)}$ is a holomorphic function which maps the disk to the right half plane. In our case, using the expression for the Schur function,
\begin{equation}\label{Caratheodory}F=\frac{\delta-\gamma}{\delta+\gamma} \end{equation}
Note that on $S^1$
$$F=\frac{(\delta_2-\gamma_2)(\delta_2+\gamma_2)^*}{(\delta_2+\gamma_2)(\delta_2+\gamma_2)^*}
=\frac{\delta_2\delta_2^*-\gamma_2\gamma_2^*+Imaginary}
{|\delta_2+\gamma_2|^2}$$
implying
$$Re(F)\frac{d\theta}{2\pi}=\frac{\prod_{k=1}^{\infty}(1-|\alpha_k|^2)}{|\gamma_2+\delta_2|^2}\frac{d\theta}{2\pi} $$
as in the Proposition.

The formula for the Caratheodory formula implies the following

\begin{lemma}\label{deltaformula}
$$\label{deltaformula}\delta=\frac 12(1+F)(\delta+\gamma) $$
This implies that the `multiplicative regular part of $g_2$', the holomorphic function
$$\left(\begin{matrix} \alpha_2&\beta_2\\\gamma_2&\delta_2\end{matrix}\right)=\left(\begin{matrix} \alpha_2&\beta_2\\\frac 12(1-F)(\delta_2+\gamma_2)&\frac 12(1+F)(\delta_2+\gamma_2)\end{matrix}\right) $$
\end{lemma}

To clarify this lemma, to guarantee that $g_2$ has a triangular factorization as in (\ref{property2}), we need to assume that $\alpha$ is absolutely summable - for our example $\alpha_n=\frac{1}{1+n\beta}$, $g_2$ does not have a triangular factorization. However the formula for $\delta$ is valid in general, as an equality in $\Delta$.

\subsection{Proof of Theorems \ref{thm1} and \ref{thm2} }

To explicitly understand the Verblunsky correspondence, one needs to calculate $F$ and $\delta_2+\gamma_2$.
In the remainder of this section we assume $\beta>0$ and set $m=1/\beta$.

\begin{theorem}Suppose that $\alpha_n=\frac{1}{n\beta+1}$. Then
$$\mu=\sqrt{\pi}2^{-m}\frac{\Gamma(m+ 1)}{\Gamma(m+\frac 12)}(1-cos(\theta)^{m}\frac{d\theta}{2\pi} $$
and hence
$$Re(F)=\sqrt{\pi}2^{-m}\frac{\Gamma(m+ 1)}{\Gamma(m+\frac 12)}(1-cos(\theta)^{m}$$
\end{theorem}

\begin{proof} For the infinite product
$$\left(\begin{matrix}\delta_2^*(z)&\gamma_2^*(z)\\\gamma_2(z)&\delta_2(z)\end{matrix}\right)=
\prod_{k=1}^{\stackrel{\infty}{\leftarrow}}\left(\begin{matrix}1&\alpha_k^*z^{-k}\\ \alpha_kz^{k} & 1 \end{matrix}\right)$$
it is straightforward to calculate that
$\gamma_2(z)=\sum_{n=1}^{\infty}\gamma_{2,n}z^n$ where
$$\gamma_{2,n}=\sum \bar{\alpha}_{i_1}\alpha_{j_1}...\bar{\alpha}_{
i_r}\alpha_{j_r}\bar{\alpha}_{i_{r+1}}$$ the sum over
multiindices satisfying
$$0<i_1<j_1<..<j_r<i_{r+1},\quad\sum i_{*}-\sum j_{*}=n,$$
and
$\delta_2(z)=1+\sum_{n=1}^{\infty}\delta_{2,n}z^n$ where
$$\delta_{2,n}=\sum \alpha_{i_1}\bar{\alpha}_{j_1}...\alpha_{i_r}\bar{\alpha}_{j_r}$$ the sum over multiindices satisfying
$$0<i_1<j_1<..<j_r,\quad\sum (j_{*}-i_{*})=n.$$
This implies that the $n$th coefficient of $\delta_2+\gamma_2$
is the left hand side of the `super-telescoping formula' (\ref{niceidentity}).
The right hand side of (\ref{niceidentity}) can be written as the `binomial coefficient' $\left(\begin{matrix}m+n-1\\n\end{matrix}\right)$
(which is implicitly defined in terms of $\Gamma$ functions). The negative binomial
formula then implies that
\begin{equation}\label{gamma+delta}\delta_2+\gamma_2=\frac{1}{(1-z)^{m}} \end{equation}
Proposition \ref{Verblunsky} now implies that
$$\mu=\Re(F)\frac{d\theta}{2\pi}=\prod_{n=1}^{\infty}(1-|\alpha_n|^2)|1-z|^{2m}\frac{d\theta}{2\pi} $$
$$=\prod_{n=1}^{\infty}(1-|\alpha_n|^2)(2(1-cos(\theta))^{m}\frac{d\theta}{2\pi} $$
Observe
$$\int_{\theta=0}^{2\pi}(1-cos(\theta))^{m}d\theta=\int_{\theta=0}^{2\pi}(2sin^2(\frac{\theta}{2}))^{m}d\theta
=2\int_{\phi=0}^{\pi}(2sin^2(\phi))^{m}d\phi$$
$$=2^{1+m}2\int_{\phi=0}^{\pi/2}sin(\phi)^{2m}d\phi=2^{1+m}\sqrt{\pi}\frac{\Gamma(m+\frac 12)}{\Gamma(m+1)} $$
Therefore, since $\mu$ is a probability measure
$$\prod_{n=1}^{\infty}(1-|\alpha_n|^2)=\sqrt{\pi}2^{-2m}\frac{\Gamma(m+ 1)}{\Gamma(m+\frac 12)}$$
hence
$$\mu=\sqrt{\pi}2^{-2m}\frac{\Gamma(m+ 1)}{\Gamma(m+\frac 12)}(2(1-cos(\theta))^{m}\frac{d\theta}{2\pi} $$
$$=\sqrt{\pi}2^{-m}\frac{\Gamma(m+ 1)}{\Gamma(m+\frac 12)}(1-cos(\theta)^{m}\frac{d\theta}{2\pi} $$
and this immediately implies the formula for $Re(F)$.

\end{proof}

To prove Theorem \ref{thm2}, we need to prove the following

\begin{theorem} In general
$$F_m(z)=1+\sqrt{\pi}2^{1-m}\frac{\Gamma(m+ 1)}{\Gamma(m+\frac 12)}\sum_{n=1}^{\infty}\sum_{k=0}^{\infty}\left(\begin{matrix}m\\n+2k
\end{matrix}\right)\left(\begin{matrix}n+2k\\k\end{matrix}\right)(\frac{-1}{2})^{n+2k}z^{n} $$
If $m$ is integral, then
$$F_m(z)=1+2\frac{(m!)^2}{(2m-1)!!}\sum_{n=1}^{m}\sum_{k=0}^{[(m-n)/2]}\frac{1}{(m-n-2k)!k!(n+k)!}
(\frac{-1}{2})^{n+2k}z^{n} $$
\end{theorem}

\begin{remark}This $g_2$ is not continuous, reflecting the fact that $\alpha_n$ is not summable. At least in
simple cases, e.g. $m=1$, it can be checked that $g_2$ does not have a triangular factorization. Hence we
cannot associate a solution to KdV, as in Remark \ref{KdV}, to this kind of loop.

\end{remark}

\begin{proof} This uses facts about Chebyshev expansions. Let $\widetilde T_0=\frac 12 T_0=\frac 12$ and $\widetilde T_n=T_n$ for $n>0$. If $r\ge 1$, then
$$x^r=2^{1-r}\sum_{k=0}^{[r/2]}\left(\begin{matrix}r\\k\end{matrix}\right)\widetilde T_{r-2k}(x) $$
The binomial formula implies
$$(1-x)^m=2^{1-r}\sum_{r\ge 0}\left(\begin{matrix}m\\r
\end{matrix}\right)(-1)^rx^r $$
$$=\sum_{r=0}^m\left(\begin{matrix}m\\r
\end{matrix}\right)(-1)^r2^{1-r}\sum_{k=0}^{[r/2]}\left(\begin{matrix}r\\k\end{matrix}\right)\widetilde T_{r-2k}(x) $$
Upon substituting $x=cos(\theta)$, this implies
$$(1-cos(\theta))^m=2\sum_{r=0}^m\sum_{k=0}^{[r/2]}\left(\begin{matrix}m\\r
\end{matrix}\right)\left(\begin{matrix}r\\k\end{matrix}\right)(\frac{-1}{2})^r\widetilde T_{r-2k}(cos(\theta)) $$
Since $T_{r-2k}(cos(\theta))=cos((r-2k)\theta)$, $r-2k\ne 0$,
$$(1-cos(\theta))^m=2Re\sum_{r=0}^m\sum_{k=0}^{[r/2]'}\left(\begin{matrix}m\\r
\end{matrix}\right)\left(\begin{matrix}r\\k\end{matrix}\right)(\frac{-1}{2})^rz^{r-2k} $$
where the prime indicates that we have to halve the sum when $r-2k=0$.
Therefore (incorporating the normalization factor for $F_m$)
$$F_m(z)=1+\sqrt{\pi}2^{1-m}\frac{\Gamma(m+ 1)}{\Gamma(m+\frac 12)}\sum\left(\begin{matrix}m\\r
\end{matrix}\right)\left(\begin{matrix}r\\k\end{matrix}\right)(\frac{-1}{2})^rz^{r-2k} $$
where the sum is over pairs $r,k\ge 0$, $r-2k>0$.

Now suppose we make the change of variables $n=r-2k$.
$$F_m(z)=1+\sqrt{\pi}2^{1-m}\frac{\Gamma(m+ 1)}{\Gamma(m+\frac 12)}\sum_{n=1}^{\infty}\sum_{k=0}^{\infty}\left(\begin{matrix}m\\n+2k
\end{matrix}\right)\left(\begin{matrix}n+2k\\k\end{matrix}\right)(\frac{-1}{2})^{n+2k}z^{n} $$

\end{proof}

\begin{examples} In case someone recognizes something!:
$$F_1(z)=1-z $$
$$F_2(z)1-\frac 43 z+\frac 13 z^2$$
$$F_3=1-\frac{15}{10}z+\frac{6}{10}z^2-\frac{1}{10}z^3 $$
$$F_4=1-\frac{56}{35}z+\frac{48}{35}z^2-\frac{8}{35}z^3+\frac{1}{36}z^4 $$
\end{examples}

\subsection{Moments of $\mu$}

To compute the moments (Fourier coefficients) of $d\mu = \frac{1}{\mathfrak Z}(1 - \cos(\theta))^m d\theta / 2\pi$ write
$$d\mu = \frac{1}{\mathfrak Z} 2^{-m} \left(2 - e^{i\theta} - e^{-i\theta}\right)^m \frac{d\theta}{2\pi}$$
Then
\begin{equation}
    c_n =\frac{1}{\mathfrak Z} 2^{-m} \int e^{-in\theta}  \left(2 - e^{i\theta} - e^{-i\theta}\right)^m \frac{d\theta}{2\pi}
\end{equation}
Expand
$$
    \left(2 - e^{i\theta} - e^{-i\theta}\right)^m = \sum_{k=0}^m  {m \choose k } 2^{m-k} (-1)^k \left(e^{i\theta} + e^{-i\theta}\right)^k = \sum_{k=0}^m \sum_{l=0}^k {m \choose k } {k \choose l } 2^{m-k} (-1)^k e^{i(k-2l) \theta}
$$
and rewrite the sum in terms of $r=k-2l$
\begin{equation}=
    \left(2 - e^{i\theta} - e^{-i\theta}\right)^m = \sum_{r=-m}^m \sum_{l=0}^{ \lfloor \frac{m-r}{2}  \rfloor} {m \choose r+2l } {r+2l \choose l } 2^{m-r-2l} (-1)^{r+2l} e^{ir \theta}
\end{equation}

This implies
\begin{theorem}
The Fourier coefficients of the measure $\frac{1}{\mathfrak Z} (1 - \cos(\theta))^m d\theta / 2\pi$ on $S^1$ are given by
\begin{equation}
    c_n =  \sqrt{\pi}2^{-m}\frac{\Gamma(m+ 1)}{\Gamma(m+\frac 12)} \sum_{l=0}^{ \lfloor \frac{m-n}{2}  \rfloor} {m \choose n+2l } {n+2l \choose l }  (\frac{-1}{2})^{n+2l}
\end{equation} $n=1,2,...$
\end{theorem}

In comparison with the other explicit examples in \cite{Simon1}, we do lack transparent expressions for the orthogonal polynomials for $\mu$.

\end{document}